\documentclass[11pt,letterpaper]{article}
\usepackage[margin=1in]{geometry}
\usepackage[T1]{fontenc}
\usepackage[utf8]{inputenc}
\usepackage{lmodern}
\usepackage{amsmath,amssymb,amsthm,mathtools}
\usepackage{microtype}
\usepackage{cite}
\usepackage{authblk}

\usepackage{enumitem}
\usepackage{needspace}
\usepackage{xcolor}
\usepackage{xurl}
\definecolor{linkblue}{RGB}{28,61,112}
\usepackage[colorlinks=true,allcolors=linkblue,bookmarksnumbered=true]{hyperref}
\hypersetup{pdftitle={An Exponential Strong Converse for Quantum Channel Discrimination},
 pdfauthor={Minbo Gao, Zhengfeng Ji, Chenghua Liu},
 pdfsubject={Channel Stein lemma, uniform Stinespring approximation, and Renyi continuity},
 pdfkeywords={quantum channel discrimination, Stein lemma, strong converse, Stinespring dilation, subchannel smoothing}}
\usepackage{aliascnt}
\usepackage[capitalise,nameinlink,noabbrev]{cleveref}
\numberwithin{equation}{section}
\newtheorem{theorem}{Theorem}[section]
\newaliascnt{lemma}{theorem}
\newtheorem{lemma}[lemma]{Lemma}
\aliascntresetthe{lemma}
\newaliascnt{proposition}{theorem}
\newtheorem{proposition}[proposition]{Proposition}
\aliascntresetthe{proposition}
\newaliascnt{corollary}{theorem}
\newtheorem{corollary}[corollary]{Corollary}
\aliascntresetthe{corollary}
\theoremstyle{remark}
\newaliascnt{remark}{theorem}

\aliascntresetthe{remark}
\DeclareMathOperator{\Tr}{Tr}
\DeclareMathOperator{\supp}{supp}
\DeclareMathOperator{\ran}{ran}
\DeclareMathOperator{\vecop}{vec}
\newcommand{\N}{\mathcal N}
\newcommand{\M}{\mathcal M}
\newcommand{\id}{\mathrm{id}}
\newcommand{\cpo}{\le_{\mathrm{CP}}}
\newcommand{\op}[1]{\left\lVert#1\right\rVert_{\infty}}
\newcommand{\trn}[1]{\left\lVert#1\right\rVert_1}
\newcommand{\dn}[1]{\left\lVert#1\right\rVert_{\diamond}}
\newcommand{\ket}[1]{\lvert#1\rangle}
\newcommand{\bra}[1]{\langle#1\rvert}
\newcommand{\Dmax}{D_{\max}}
\newcommand{\binary}{h_2}
\setlist[enumerate]{itemsep=2pt,topsep=4pt,leftmargin=*}
\makeatletter
\renewcommand{\@maketitle}{%
 \begin{center}
  {\fontsize{17}{21}\selectfont\bfseries\@title\par}
  \ifx\@author\@empty\else\vspace{1em}
   {\normalsize\lineskip .5em\@author\par}\fi
  \ifx\@date\@empty\else\vspace{.7em}{\normalsize\@date\par}\fi
 \end{center}\vspace{.6em}}
\makeatother
\title{Quantum Channel Stein's Lemma with an Exponential Strong Converse}
\author[1,2]{Minbo Gao\thanks{\href{mailto:gmb17@tsinghua.org.cn}{\texttt{gmb17@tsinghua.org.cn}}}}
\author[3]{Zhengfeng Ji\thanks{\href{mailto:jizhengfeng@tsinghua.edu.cn}{\texttt{jizhengfeng@tsinghua.edu.cn}}}}
\author[1,2]{Chenghua Liu\thanks{\href{mailto:liuch.russell@gmail.com}{\texttt{liuch.russell@gmail.com}}}}
\affil[1]{Institute of Software, Chinese Academy of Sciences, Beijing, China}
\affil[2]{University of Chinese Academy of Sciences, Beijing, China}
\affil[3]{Department of Computer Science and Technology, Tsinghua University, Beijing, China}
\date{}
\begin{document}
\maketitle

\begin{abstract}
  We establish the quantum channel Stein lemma at every fixed type-I error
  tolerance in $(0,1)$ and prove an exponential strong converse under arbitrary
  adaptive strategies.
  For any finite-dimensional channel pair, the optimal type-II exponent equals
  the regularized channel relative entropy and is attained by parallel
  strategies.
  When this divergence is finite, every strictly larger type-II rate forces the
  probability of accepting the null hypothesis to decay exponentially, uniformly
  over strategies and quantum memories.
  The proof rests on a uniform Stinespring approximation: tensor powers of the
  null dilation are approximated by tensor powers of the alternative dilation
  after applying an auxiliary linear map to the environment.
  The operator-norm error is exponentially small under a squared-norm bound at
  any rate above the regularized relative entropy.
  We construct the approximation from a weak testing bound by iterative rate
  reduction and a fixed-block tensor expansion with an exact residual.
  It implies right continuity at order one of the regularized sandwiched R\'enyi
  channel divergence, which yields the adaptive converse through the channel
  chain rule.
  We also obtain a sharp hockey-stick threshold and a subchannel smoothing
  asymptotic equipartition property for fixed and subexponentially vanishing
  diamond-norm errors.
\end{abstract}

\smallskip

\section{Introduction}\label{sec:introduction}

Stein's lemma identifies relative entropy as the optimal type-II error exponent
in asymmetric hypothesis testing~\cite[Sec.~11.8]{CoverThomas}.
It is a basic result in the classical theory of error exponents~\cite{Blahut}.
Hiai and Petz established the quantum state version~\cite{HiaiPetz}, and Ogawa
and Nagaoka proved the strong converse~\cite{OgawaNagaoka}, giving the same
exponent at every fixed type-I error tolerance in $(0,1)$.
Later work determined the strong-converse exponent~\cite{MosonyiOgawa} and the
second-order asymptotics~\cite{Li,TomamichelHayashi}.
These results concern independent copies of a fixed pair of states, with
optimization over the final measurement.

Channel discrimination also optimizes how the states to be measured are
produced.
Even for a memoryless unknown channel, the experiment may use inputs entangled
across all channel uses and adaptive processing with a quantum memory.
The resulting output states need not be tensor powers of any fixed pair.
A converse must therefore control entire sequences of experiments, rather than
the repetition of a single state discrimination problem.

For channels $\N$ and $\M$, the natural candidate for the Stein exponent is the
regularized channel relative entropy
\begin{equation}
 d=D^\infty(\N\Vert\M)
   =\sup_{k\ge1}\frac1k D(\N^{\otimes k}\Vert\M^{\otimes k}),
 \label{eq:intro-d}
\end{equation}
where $D(\N\Vert\M)$ is the largest output relative entropy obtained from a
common input and an arbitrary reference.
Every rate below $d$ is achievable: choose a suitable input block, repeat it
independently, and apply the state Stein lemma.
The difficulty is the fixed-error converse.
If $\beta_n^*(\varepsilon)$ is the optimal parallel type-II error at type-I
error at most $\varepsilon$, binary data processing gives only
\begin{equation}
 -\frac1n\log\beta_n^*(\varepsilon)
 \le\frac{d+1/n}{1-\varepsilon}.
 \label{eq:intro-weak-converse}
\end{equation}
For $0<d<\infty$, this leaves a gap at every fixed $\varepsilon>0$.
Adaptive strategies add a further difficulty: the input at a given use can
already differ under the two hypotheses.

We prove that $d$ is the fixed-error Stein exponent for arbitrary
finite-dimensional channel pairs, under both parallel and general adaptive
strategies.
For finite $d$, we also establish an exponential strong converse: for every
$R>d$, any adaptive $n$-use test with type-II error at most $2^{-nR}$ accepts
the null hypothesis with probability at most $2^{-nc_R}$, where $c_R>0$ is
independent of the strategy and its memory dimensions.
Thus neither adaptive control nor the choice of a fixed error tolerance changes
the first-order exponent.
The precise statement is given in \cref{thm:adaptive}.

\paragraph{Related work.}
Hayashi studied classical channel discrimination with adaptive input
selection~\cite{Hayashi}.
For quantum channels, Cooney, Mosonyi, and Wilde determined strong-converse
exponents when the alternative is a replacer channel~\cite{CooneyMosonyiWilde}.
Wilde, Berta, Hirche, and Kaur developed amortized channel divergences for
adaptive converses and proved a strong Stein lemma for classical--quantum
channels~\cite{WildeEtAl}.
Wang and Wilde formulated the resource theory of asymmetric channel
distinguishability and established its weak Stein
characterization~\cite{WangWilde}.

For general channel pairs, Fang, Fawzi, Renner, and Sutter proved a
relative-entropy chain rule that identifies amortization with regularization and
determines the weak adaptive Stein exponent~\cite{FangEtAl}.
Bergh, Datta, Salzmann, and Wilde gave quantitative comparisons between adaptive
and parallel discrimination strategies for finitely many channel
uses~\cite{BerghEtAl}.
Fawzi and Fawzi established the corresponding sandwiched R\'enyi chain rule and
characterized the optimal adaptive strong-converse exponent~\cite{FawziFawzi}.
The remaining question was whether that exponent is strictly positive at every
rate above $d$.
As explained in~\cite[Remark~5.6]{FawziFawzi}, locating this threshold requires
right continuity at order one of the regularized sandwiched R\'enyi channel
divergence.
Fang, Gour, and Wang further developed the connections among weak Stein results,
exponential strong converses, continuity, and asymptotic equipartition
properties formulated through output-state smoothing~\cite{FangGourWang}.

Smoothing provides another formulation of the threshold problem.
For states, smooth entropies and max-relative entropy have asymptotic limits
governed by von Neumann and relative entropies~\cite{Datta,TCR,DMHB}.
For channels, the normalization of the approximating map matters.
Gour related subchannel smoothing to the parallel Stein strong converse and
sharp testing thresholds, while showing that smoothing over trace-preserving
channels can have a strictly larger asymptotic rate~\cite{Gour}.
That work also develops uniform completely positive approximations from testing
information.
Our smoothing result concerns subchannels and does not impose exact trace
preservation on the approximating maps.

\paragraph{Our approach.}
The central step is a uniform comparison of Stinespring dilations.
We approximate tensor powers of the null dilation by an auxiliary linear map on the
alternative environment, with exponentially small error and a squared-norm rate
arbitrarily close to $d$ from above.
The same map works for every input and reference system.
This is the additional control needed to pass from a weak testing estimate to
the fixed-error threshold: it bounds acceptance probabilities directly, proves
the R\'enyi continuity required by the adaptive chain rule, and yields dominated
subchannel approximations.
\Cref{sec:approx-statement} states and explains the approximation theorem, and
\cref{sec:proof-strategy} outlines its proof.

The results and proof overview appear in \cref{sec:setup}.
We derive the Stein theorem from the approximation in \cref{sec:main-proof},
then prove the approximation itself in \cref{sec:approximation}.
\Cref{sec:consequences} treats testing thresholds and subchannel smoothing, and
\cref{sec:discussion} discusses further questions.

\section{Main Results and Proof Overview}\label{sec:setup}

\subsection{Notation and discrimination model}\label{sec:model}

All Hilbert spaces are finite dimensional, and $n$ denotes a positive
integer. We write $\mathcal L(A,B)$ for linear operators from $A$ to
$B$, $\mathcal L(A)=\mathcal L(A,A)$, and $\mathcal S(A)$ for density
operators on $A$. A channel $\N:A\to B$ is a completely positive
trace-preserving (CPTP) map from $\mathcal L(A)$ to $\mathcal L(B)$.
For completely positive (CP) maps, $\Phi\cpo\Psi$ means that
$\Psi-\Phi$ is CP. We write $A^n=A^{\otimes n}$, identify tensor
products under the canonical permutations of their factors, and omit
identity subscripts only when the system is unambiguous.

The identity operator and identity channel on $A$ are $I_A$ and
$\id_A$. The adjoint, transpose in a fixed basis, and Moore--Penrose
inverse are $X^\dagger$, $X^{\mathsf T}$, and $X^+$, respectively.
For a subspace $\mathsf K$, $\Pi_{\mathsf K}$ is its orthogonal
projection. For Hermitian $X$, $X_+$ denotes its positive part. The operator and
trace norms are $\op X$ and $\trn X$; for vectors we use the Euclidean
norm $\|x\|$. Logarithms are base two
unless denoted by $\ln$, and $0\log0=0$. The binary entropy is
$\binary(x)=-x\log x-(1-x)\log(1-x)$ for $x\in[0,1]$.
We refer to~\cite{Watrous} for standard facts about quantum channels
and matrix norms.

For states $\rho,\sigma$, their relative entropy is
$D(\rho\Vert\sigma)=\Tr\rho(\log\rho-\log\sigma)$ when
$\supp\rho\subseteq\supp\sigma$, and is $+\infty$ otherwise.
An effect $0\le T\le I$ represents acceptance of the null hypothesis.
Its acceptance probabilities under $\rho$ and $\sigma$ are
$p=\Tr T\rho$ and $q=\Tr T\sigma$, respectively; the type-I and
type-II errors are $1-p$ and $q$. For $0<\varepsilon<1$, define
\begin{equation*}
 D_H^\varepsilon(\rho\Vert\sigma)
 =-\log\inf_{0\le T\le I}
       \{\Tr T\sigma:\Tr T\rho\ge1-\varepsilon\}.
\end{equation*}

A parallel $n$-use strategy may be taken to prepare a pure state
$\ket\psi_{RA^n}$ and measure $RB^n$ after the unknown channel has
acted on all inputs.
The two output states are
\begin{equation}
 \rho_n^\psi=(\id_R\otimes\N^{\otimes n})(\ket\psi\bra\psi),\qquad
 \sigma_n^\psi=(\id_R\otimes\M^{\otimes n})(\ket\psi\bra\psi).
 \label{eq:outputs}
\end{equation}
Purifying an input cannot reduce the divergences used here or worsen a
test, since
the purifying system may be ignored at the final measurement.
A pure input has Schmidt rank at most $\dim A^n$ across the
reference--input cut, so $R\simeq A^n$ suffices up to a reference
isometry. We use this reduction in all parallel optimizations below,
without restricting entanglement across channel uses.

We use the same divergence symbol for states and channels; the
arguments distinguish them. For a fixed reference $R\simeq A$,
optimize over unit vectors $\ket\psi_{RA}\in R\otimes A$:
\begin{equation}
 D(\N\Vert\M)
 =\sup_{\ket\psi_{RA}}D\bigl((\id_R\otimes\N)(\ket\psi\bra\psi)
                 \Vert(\id_R\otimes\M)(\ket\psi\bra\psi)\bigr).
 \label{eq:channel-d}
\end{equation}
Set $a_n=D(\N^{\otimes n}\Vert\M^{\otimes n})$. Product inputs and
additivity of state relative entropy give $a_{n+m}\ge a_n+a_m$;
hence Fekete's lemma yields
\begin{equation}
 d=D^\infty(\N\Vert\M)
  =\lim_{n\to\infty}\frac{a_n}{n}
  =\sup_{n\ge1}\frac{a_n}{n},
 \label{eq:d-setup}
\end{equation}
where $+\infty$ is allowed. The optimal parallel type-II error is
\begin{equation}
 \beta_n^*(\varepsilon)
 =\inf_{\ket\psi,\,0\le T\le I}
   \{\Tr T\sigma_n^\psi:\Tr T\rho_n^\psi\ge1-\varepsilon\}.
 \label{eq:beta-n}
\end{equation}
Here $R\simeq A^n$ is fixed, $\ket\psi\in R\otimes A^n$ is a unit vector,
and $T$ acts on $R\otimes B^n$. The feasible set is compact,
so the infimum is attained.

An adaptive strategy alternates the unknown channel with arbitrary
CPTP maps on the preceding output and a quantum memory. Its initial
state and intermediate maps are the same under both hypotheses.
Intermediate measurements and classical feedback are included by
storing outcomes in the memory. Each memory is finite dimensional,
but its dimension is unrestricted and may depend on $n$.
We write $\beta_{n,\mathrm{ad}}^*(\varepsilon)$ for the corresponding
infimum of the type-II error under the same type-I constraint.
This is the standard adaptive model used in~\cite{WildeEtAl,WangWilde}.
For either strategy class, $p$ and $q$ denote the probabilities of
accepting the null hypothesis when the channel is $\N$ and $\M$,
respectively.

\subsection{The Stein threshold}\label{sec:stein-statement}

The first theorem determines the fixed-error exponent and gives an
exponential bound above it. The bound applies to every adaptive strategy,
including those with memory dimensions that grow with the blocklength.

\begin{theorem}[Channel Stein lemma and exponential strong converse]
\label{thm:adaptive}
Let $\N,\M:A\to B$ be channels and set $d=D^\infty(\N\Vert\M)$.
For every $0<\varepsilon<1$,
\begin{equation}
 \lim_{n\to\infty}-\frac1n\log\beta_{n,\mathrm{ad}}^*(\varepsilon)
 =\lim_{n\to\infty}-\frac1n\log\beta_n^*(\varepsilon)=d.
 \label{eq:adaptive-stein}
\end{equation}
If $d<\infty$, then for every $R>d$ there is $c_R>0$ such that
all adaptive $n$-use tests satisfy
\begin{equation}
 q\le2^{-nR}\quad\Longrightarrow\quad p\le2^{-nc_R}
 \qquad(n\ge1).
 \label{eq:adaptive-power}
\end{equation}
\end{theorem}

In \eqref{eq:adaptive-power}, $c_R$ depends only on the channels and
$R$. Every nonnegative rate below $d$ is attained by repeating a suitable
fixed input block, with type-I error tending to zero. Above $d$, the
probability of accepting the null hypothesis tends to zero exponentially.
The fixed-error threshold is consequently independent of both the chosen
tolerance and access to adaptive control.

The limits in \eqref{eq:adaptive-stein} are understood in the extended
real sense. When $d=+\infty$, parallel tests have zero type-II error
at all sufficiently large blocklengths; \cref{prop:infinite} gives
an explicit construction. For replacer channels $\N(X)=\Tr(X)\rho$ and
$\M(X)=\Tr(X)\sigma$, the threshold is $D(\rho\Vert\sigma)$.
If $\N=\M$, then $p=q$ and $\beta_n^*(\varepsilon)=1-\varepsilon$.

\subsection{Uniform Stinespring approximation}\label{sec:approx-statement}

Our main technical result concerns the dilations of the two channels.
A Stinespring isometry $V_N:A\to B\otimes\mathsf E_N$ represents $\N$
as $\N(X)=\Tr_{\mathsf E_N}V_NXV_N^\dagger$~\cite{Stinespring};
choose an analogous dilation $V_M$ of $\M$.
Rather than compare the output states separately for each input, we
approximate the dilation itself. Specifically, we seek an auxiliary linear map
$C_n$ between the two environment spaces such that
$(I_{B^n}\otimes C_n)V_M^{\otimes n}$ is close to $V_N^{\otimes n}$.
The following theorem controls both the approximation error and the
norm of this map.

\begin{theorem}[Exponential Stinespring approximation]
\label{thm:approximation}
Let $\N,\M:A\to B$ be channels with $d=D^\infty(\N\Vert\M)<\infty$,
and fix Stinespring isometries $V_N:A\to B\otimes\mathsf E_N$ and
$V_M:A\to B\otimes\mathsf E_M$. For every $S>d$, there are
$K,\gamma>0$ and auxiliary linear maps
$C_n:\mathsf E_M^{\otimes n}\to\mathsf E_N^{\otimes n}$,
for all sufficiently large $n$, such that $\op{C_n}\le2^{nS/2}$ and
\begin{equation}
 \op{V_N^{\otimes n}-(I_{B^n}\otimes C_n)V_M^{\otimes n}}
 \le K e^{-\gamma n}.
 \label{eq:strong-error}
\end{equation}
\end{theorem}

Write $\widehat V_n=(I_{B^n}\otimes C_n)V_M^{\otimes n}$ and
$e_n=\op{V_N^{\otimes n}-\widehat V_n}$. For any unit vector on
$R\otimes A^n$, the two operators give output vectors at Euclidean
distance at most $e_n$, because tensoring with $I_R$ preserves the
operator norm. The same auxiliary map therefore works uniformly over all
inputs and reference dimensions.

The norm constraint determines how this approximation can be used in a
test. Applying $C_n$ can increase an acceptance amplitude by a factor of at
most $\op{C_n}$, so its squared norm is the corresponding probability
factor. This explains the normalization $\op{C_n}^2\le2^{nS}$.
 Stinespring continuity estimates~\cite{KSW} provide related comparisons
 between dilation and channel distance. Here we use the direct uniform
 output bound above, because $C_n$ is an auxiliary map
 whose norm may grow exponentially with $n$.

Finite divergence already gives an exact comparison
$V_N=(I_B\otimes C_*)V_M$ with $\op{C_*}\le2^{L/2}$ at some finite
domination rate $L\ge d$, where $L=\log c$ for any finite $c\ge1$
with $\N\cpo c\M$; see \cref{lem:finite,cor:exact}. The gain in
\cref{thm:approximation} is that allowing an exponentially small error
reduces the required rate to any $S>d$, including $S<L$.
Moreover, $C_n$ and $C_*^{\otimes n}$ act between the same prescribed
environments. Their difference is therefore an exact auxiliary map for
the residual $V_N^{\otimes n}-\widehat V_n$. This fact is used both in
the tensor construction and in the R\'enyi estimates.

The approximation also controls the channel after the environment is
discarded. Dividing $\widehat V_n$ by $1+e_n$ makes it a contraction,
which defines a completely positive, trace-nonincreasing map, or
subchannel. This map is exponentially close to $\N^{\otimes n}$ in
diamond norm and is dominated, in completely positive order, by
$2^{nS}\M^{\otimes n}$. Thus approximation and domination can be
imposed simultaneously at any rate above $d$. We prove this consequence
in \cref{sec:subchannels} and use it to determine the subchannel smoothing
rate. We discuss possible quantitative refinements of this approximation
in \cref{sec:discussion}.

\subsection{Overview of the proof}\label{sec:proof-strategy}

We first derive the statistical results from the approximation, then
explain its construction. The construction uses only a weak testing bound.

\paragraph{From approximation to testing.}
For a parallel input and a binary test, the approximation in
\cref{thm:approximation} gives
\begin{equation*}
 \sqrt p\le\op{C_n}\sqrt q+e_n.
\end{equation*}
Indeed, the test acts on the output and reference, whereas $C_n$ acts
on the environment. If $d<S<R$ and $q\le2^{-nR}$, this bounds
$\sqrt p$ by $2^{-n(R-S)/2}+Ke^{-\gamma n}$. The parallel strong
converse follows, uniformly over all inputs and tests.

For adaptive strategies, we use the approximation to prove R\'enyi
continuity. Each null output purification is the sum of an approximate
part with cost rate $S$ and a small residual with a finite
cost rate $L$. A weighted Schatten-norm estimate then bounds the
regularized R\'enyi divergence by $S$ for orders sufficiently close to
one. Since $S>d$ is arbitrary, this proves continuity at order one. The Fawzi--Fawzi chain rule~\cite{FawziFawzi}
then bounds the divergence increase at each adaptive use. Choosing the
order sufficiently close to one gives a positive converse exponent at
every $R>d$.

\paragraph{Constructing the approximation.}
For any $r>d$, relative-entropy data processing bounds the supremum of
$p-2^{nr}q$ over parallel inputs and tests by $d/r+o(1)$.
Semidefinite duality converts this into a uniform approximation using
an auxiliary map of norm at most $2^{nr/2}$ and with error at most a constant
$\delta<1$. We also have the exact auxiliary map at rate $L$. The task is
to obtain the accuracy of the exact representation at a rate close to
$r$.

Directly repeating the low-rate approximation can accumulate error.
Instead, let $X$ be such an approximation and $Y$ an accurate one at a
higher rate. Write $Y=X+H$ and expand $(X+H)^{\otimes m}$. Keeping
only terms with at most a suitable fraction of $H$ factors lowers the
cost rate. Since $\op H<1$, choosing this fraction sufficiently
close to one makes the omitted tail exponentially small in $m$.
The internal blocklength is also allowed to grow, which controls the
error in replacing $Y^{\otimes m}$ by the target. A fixed number of
iterations gives a stretched-exponential error at every rate above $r$;
this is \cref{lem:amplification}.

We then choose a sufficiently accurate block $X$ of length $k$ and keep $k$
fixed.
Its exact residual $H=V_N^{\otimes k}-X$ satisfies
$(X+H)^{\otimes m}=V_N^{\otimes km}$.
Because $\op H$ can be made arbitrarily small, using only a small fraction of
residual factors achieves exponential accuracy with an arbitrarily small
increase in rate.
There is now no error from repeating an approximate target.
\Cref{lem:exponential} formalizes this second construction.

\section{Proof of the Channel Stein Theorem}\label{sec:main-proof}

We use \cref{thm:approximation} to prove the parallel converse,
R\'enyi continuity, and the adaptive converse, in that order. We first
record the algebraic facts needed in these deductions. The approximation
theorem is proved independently in \cref{sec:approximation}, using only
a weak testing bound.

\subsection{Support and exact environmental comparison}\label{sec:algebra}

For a linear map $\Phi:A\to B$, its unnormalized Choi operator is
$J_\Phi=(\id_A\otimes\Phi)(\ket\Omega\bra\Omega)$, where
$\ket\Omega=\sum_{i=1}^{\dim A}\ket i\ket i$ and the first copy of
$A$ is a reference. The Choi criterion states that $\Phi$ is CP
exactly when $J_\Phi\ge0$; a channel also satisfies
$\Tr_BJ_\Phi=I_A$~\cite{Choi}. This turns finite divergence into an
operator-order condition.

\begin{lemma}[Finite domination and regularization]\label{lem:finite}
For channels $\N,\M:A\to B$, finiteness of $D^\infty(\N\Vert\M)$
is equivalent to $\supp J_\N\subseteq\supp J_\M$, and to
$\N\cpo c\M$ for some finite $c\ge1$.
\end{lemma}

Under these equivalent conditions, the proof gives the explicit choice
$c=\op{(J_\M^+)^{1/2}J_\N(J_\M^+)^{1/2}}$.
We write $L=\log c$ for this domination rate; it satisfies
$0\le a_n\le nd\le nL$. If support inclusion fails, then
$a_n=+\infty$ for every $n$.

\begin{proof}
  The Choi criterion and finite-dimensional matrix order identify support
  inclusion with domination by a finite multiple.
  Under support inclusion, the stated value of $c$ gives $J_\N\le cJ_\M$; taking
  traces shows $c\ge1$.
  To see why this domination persists under tensor powers, let
  $\Phi\cpo\Psi$ be CP maps. The identity
  \begin{equation*}
    \Psi^{\otimes(n+1)}-\Phi^{\otimes(n+1)}
    =(\Psi^{\otimes n}-\Phi^{\otimes n})\otimes\Psi
      +\Phi^{\otimes n}\otimes(\Psi-\Phi)
  \end{equation*}
  shows the induction step: if the difference at $n$ is CP, then both
  terms on the right are CP, since tensor products of CP maps are CP.
  Starting with $\Psi-\Phi$ therefore proves
  $\Phi^{\otimes n}\cpo\Psi^{\otimes n}$ for every $n$.
  Applying this with $\Phi=\N$ and $\Psi=c\M$ gives
  $\N^{\otimes n}\cpo c^n\M^{\otimes n}$.

  Complete positivity allows us to tensor this difference with $\id_R$
  and apply it to any input $\ket\psi\bra\psi$. Thus every output pair
  satisfies $\rho_n^\psi\le c^n\sigma_n^\psi$, which also implies
  $\supp\rho_n^\psi\subseteq\supp\sigma_n^\psi$.
  Work on $\supp\sigma_n^\psi$ and let $I$ denote its identity operator.
  Adding $\eta I$ for $\eta>0$ makes both states strictly positive on
  this subspace. Since $c\ge1$, we have
  \begin{equation*}
    \rho_n^\psi+\eta I
    \le c^n\sigma_n^\psi+\eta I
    \le c^n(\sigma_n^\psi+\eta I).
  \end{equation*}
  Operator monotonicity of the logarithm and taking the trace against
  $\rho_n^\psi$ give
  \begin{equation*}
    \Tr\rho_n^\psi\bigl[\log(\rho_n^\psi+\eta I)
    -\log(\sigma_n^\psi+\eta I)\bigr]\le n\log c.
  \end{equation*}
  Letting $\eta\downarrow0$ yields
  $D(\rho_n^\psi\Vert\sigma_n^\psi)\le n\log c$ for every input
  $\ket\psi$. Taking the supremum over inputs gives
  $a_n\le n\log c=nL$.

  As noted in \eqref{eq:d-setup}, superadditivity implies $d=\sup_n a_n/n$;
  nonnegativity now gives all stated bounds.
  Conversely, a normalized maximally entangled input produces $J_\N/\dim A$ and
  $J_\M/\dim A$.
  If support inclusion fails, these states have infinite relative entropy, as do
  their tensor powers.
  Hence $a_n=+\infty$ for every $n$, excluding $d<\infty$.
\end{proof}

To lift Choi domination to fixed dilations, we use the following
finite-dimensional form of Douglas factorization~\cite{Douglas}.

\begin{lemma}[Matrix factorization]\label{lem:factor}
Let $F:\mathsf E\to\mathsf H$ and $G:\mathsf F\to\mathsf H$ be
linear maps. If
$FF^\dagger\le tGG^\dagger$ for $t\ge0$, then $F=GD$ for a linear
map $D:\mathsf E\to\mathsf F$ with $\op D\le\sqrt t$.
\end{lemma}

\begin{proof}
The order assumption implies $\ker G^\dagger\subseteq\ker F^\dagger$,
hence $\ran F\subseteq\ran G$. Set $D=G^+F$, so $GD=F$.
Congruence of the assumed inequality by $G^+$ yields
$DD^\dagger\le tG^+GG^\dagger(G^+)^\dagger
=t\Pi_{\ran G^\dagger}\le tI_{\mathsf F}$.
Taking the operator norm proves the claim, including $t=0$.
\end{proof}

Fix Stinespring isometries $V_N:A\to B\otimes\mathsf E_N$ and
$V_M:A\to B\otimes\mathsf E_M$ for the channel pair. Choose
orthonormal environment bases, write
$V_N\ket x=\sum_{j=1}^u N_j\ket x\otimes\ket j$ with $u=\dim\mathsf E_N$,
and similarly use $v=\dim\mathsf E_M$ Kraus operators for $V_M$.
For $K:A\to B$, define
$\vecop(K)=\sum_i\ket i\otimes K\ket i$ and let
$F_N=[\vecop(N_1)\ \cdots\ \vecop(N_u)]$, with $F_M$ defined in
the same way. Then $F_NF_N^\dagger=J_\N$ and
$F_MF_M^\dagger=J_\M$.

If $F_0=F_MD$, its corresponding dilation operator is
$(I_B\otimes D^{\mathsf T})V_M$: the transpose is required by this
column convention. Transposition preserves the operator norm even for
rectangular matrices. More generally, reshaping
$V\ket x=\sum_j K_j\ket x\otimes\ket j$ into the column matrix
$F=[\vecop(K_1)\ \cdots]$ gives
\begin{equation}
 \Tr_BFF^\dagger=(V^\dagger V)^{\mathsf T},\qquad
 \op V^2=\op{\Tr_BFF^\dagger}.
 \label{eq:reshape-norm}
\end{equation}
Indeed, the $(i,j)$ entry of the partial trace is
$\sum_a\langle K_a j,K_a i\rangle$, which is the $(j,i)$ entry of
$\sum_aK_a^\dagger K_a$. The identity does not require $V$ to be an
isometry.

The factorization lemma now provides an exact comparison in the
prescribed environments. This representation will also control the
residuals of approximate comparisons.

\begin{corollary}[Exact auxiliary map]\label{cor:exact}
Fix Stinespring isometries $V_N:A\to B\otimes\mathsf E_N$ and
$V_M:A\to B\otimes\mathsf E_M$ for channels $\N,\M$.
If $\N\cpo c\M$ for $c\ge1$, there is an auxiliary linear map
$C_*:\mathsf E_M\to\mathsf E_N$ with $\op{C_*}\le\sqrt c$
and $V_N=(I_B\otimes C_*)V_M$.
\end{corollary}

\begin{proof}
Apply \cref{lem:factor} to
$F_NF_N^\dagger\le cF_MF_M^\dagger$. It gives $F_N=F_MD$ with
$D:\mathsf E_N\to\mathsf E_M$ and $\op D\le\sqrt c$.
The column convention above gives the claimed identity with
$C_*=D^{\mathsf T}:\mathsf E_M\to\mathsf E_N$.
\end{proof}

\subsection{The parallel converse and achievability}\label{sec:parallel-proof}

An operator-norm approximation controls every input and test through
the following acceptance-amplitude bound.

\begin{lemma}[Acceptance amplitudes]\label{lem:amplitude}
Fix Stinespring isometries $V_N:A\to B\otimes\mathsf E_N$ and
$V_M:A\to B\otimes\mathsf E_M$ for $\N,\M$.
For $C:\mathsf E_M\to\mathsf E_N$, set
$e=\op{V_N-(I_B\otimes C)V_M}$. For any reference-assisted input
and binary test, let $p,q$ be its acceptance probabilities under
$\N,\M$, respectively. Then
\begin{equation*}
 \sqrt p\le\op C\sqrt q+e.
\end{equation*}
\end{lemma}

\begin{proof}
After purifying the input and enlarging $R$, it suffices to consider
a unit vector $\ket\psi_{RA}$ and an effect $0\le T\le I_{RB}$.
Set
$x=(I_R\otimes V_N)\ket\psi$ and
$y=(I_R\otimes V_M)\ket\psi$. Tensoring the error with $I_R$
gives $\|x-(I_{RB}\otimes C)y\|\le e$. As $\sqrt T$ is a
contraction and acts on systems disjoint from $C$,
\begin{align*}
 \sqrt p
 &=\|(\sqrt T\otimes I_{\mathsf E_N})x\|\\
 &\le\|(I_{RB}\otimes C)(\sqrt T\otimes I_{\mathsf E_M})y\|+e
 \le\op C\sqrt q+e.
\end{align*}
The argument is independent of the size of the reference.
\end{proof}

Applying this lemma to tensor-power channels gives the parallel
converse. The matching lower bound requires only a fixed input block
and the state Stein lemma.

\begin{proposition}[Parallel Stein lemma and exponential converse]
\label{prop:parallel}
Let $\N,\M:A\to B$ be channels with $d=D^\infty(\N\Vert\M)<\infty$.
For every $0<\varepsilon<1$,
$\lim_{n\to\infty}-n^{-1}\log\beta_n^*(\varepsilon)=d$.
For every $R>d$, there are $K_R,\gamma_R>0$ such that all parallel
$n$-use tests satisfy
\begin{equation}
 q\le2^{-nR}\quad\Longrightarrow\quad p\le K_R e^{-\gamma_R n}
 \label{eq:main-power}
\end{equation}
for all sufficiently large $n$.
\end{proposition}

\begin{proof}
Fix $d<S<R$, and take $C_n,e_n$ from
\cref{thm:approximation}. \cref{lem:amplitude} applies
to the fixed tensor-power dilations. When $q\le2^{-nR}$, it gives
\begin{equation*}
 p\le\left(2^{-n(R-S)/2}+Ke^{-\gamma n}\right)^2
 \le(1+K)^2e^{-2\min\{(R-S)\ln2/2,\gamma\}n}.
\end{equation*}
This proves \eqref{eq:main-power}. For fixed $\varepsilon$, its
right side is eventually smaller than $1-\varepsilon$, so every
feasible test in \eqref{eq:beta-n} has $q>2^{-nR}$.
Thus $\beta_n^*(\varepsilon)\ge2^{-nR}$ eventually. Since $R>d$
is arbitrary, the limit superior of the normalized negative logarithm
is at most $d$.

For the lower bound, fix a blocklength $k$ and a pure input for $k$
uses, with output states $\rho,\sigma$. Finite domination ensures
support inclusion. The state Stein lemma gives
\begin{equation*}
 \lim_{m\to\infty}\frac1m
 D_H^\varepsilon(\rho^{\otimes m}\Vert\sigma^{\otimes m})
 =D(\rho\Vert\sigma).
\end{equation*}
Its direct part also attains each strictly smaller rate with null
acceptance tending to one~\cite{HiaiPetz,OgawaNagaoka}.
Write $n=mk+\ell$ with $0\le\ell<k$. Repeat the input block $m$
times, use any fixed state on the remaining inputs, and ignore the
remaining outputs. This is an admissible parallel strategy, so
\begin{equation}
 -\log\beta_n^*(\varepsilon)
 \ge D_H^\varepsilon(\rho^{\otimes m}\Vert\sigma^{\otimes m}).
 \label{eq:block-direct}
\end{equation}
Since $m/n\to1/k$, the limit inferior is at least
$D(\rho\Vert\sigma)/k$. Taking the supremum over the input and then
over $k$ gives $\sup_k a_k/k=d$. This proves the limit and also
justifies the below-threshold achievability assertion following
\cref{thm:adaptive}.
\end{proof}

If Choi support inclusion fails, a zero-probability event under the
alternative gives a direct test.

\begin{proposition}[Failure of Choi support inclusion]\label{prop:infinite}
Let $\N,\M:A\to B$ be channels with
$\supp J_\N\not\subseteq\supp J_\M$. There is
$\lambda\in(0,1]$ such that, for every $n\ge1$, a parallel
$n$-use test has $q_n=0$ and $p_n=1-(1-\lambda)^n$.
\end{proposition}

Consequently, $\beta_n^*(\varepsilon)=0$ for all sufficiently large
$n$ at every fixed $0<\varepsilon<1$.

\begin{proof}
A normalized maximally entangled input gives the states
$\rho=J_\N/\dim A$ and $\sigma=J_\M/\dim A$.
Let $Q=\Pi_{\ker\sigma}$. Support failure implies
$\lambda=\Tr Q\rho>0$, while $\Tr Q\sigma=0$.
Repeat the input independently, measure $\{Q,I-Q\}$ on each output
and reference, and accept the null if at least one outcome is $Q$.
The stated probabilities follow by independence. Its type-I error
$(1-\lambda)^n$ eventually lies below every fixed positive tolerance.
\end{proof}

\subsection{R\'enyi continuity}\label{sec:renyi}

To handle adaptive strategies, we need a divergence with a chain rule for
different input states under the two hypotheses.
We use the sandwiched R\'enyi divergence, introduced
in~\cite{MullerLennert,WWY}.
For $\alpha>1$, a state $\sigma$, and a positive operator $\tau$ supported on
$\supp\sigma$, define
\begin{equation}
  \widetilde Q_\alpha(\tau\Vert\sigma) =
  \Tr\left(\sigma^{-\frac{\alpha-1}{2\alpha}}
    \tau\sigma^{-\frac{\alpha-1}{2\alpha}}\right)^\alpha.
 \label{eq:sandwiched-q}
\end{equation}
Here inverse powers are taken on $\supp\sigma$, and the trace is
evaluated on that subspace.
For a state $\rho$, set
$\widetilde D_\alpha(\rho\Vert\sigma) = (\alpha-1)^{-1}\log\widetilde Q_\alpha(\rho\Vert\sigma)$,
with value $+\infty$ without support inclusion.
We use data processing~\cite{Beigi,FrankLieb} and
$\widetilde D_\alpha\ge D$~\cite{MullerLennert}.
For a rectangular matrix $F$, its Schatten norm is
$\|F\|_p=(\Tr(F^\dagger F)^{p/2})^{1/p}$, for $p\ge1$.

Define $\widetilde D_\alpha(\N\Vert\M)$ by the same reference-assisted
optimization as \eqref{eq:channel-d}, with $D$ replaced by
$\widetilde D_\alpha$, and let
\begin{equation*}
 d_\alpha:=\widetilde D_\alpha^\infty(\N\Vert\M)
 =\lim_{n\to\infty}\frac1n
       \widetilde D_\alpha(\N^{\otimes n}\Vert\M^{\otimes n}).
\end{equation*}
For $d<\infty$ and $1<\alpha\le2$, this limit equals the supremum
of the normalized block values and lies in $[d,L]$, where $L$ is the
domination rate of \cref{lem:finite}. Indeed, superadditivity
follows from product inputs, and the upper bound follows from the
small-trace estimate below applied to $\rho_n^\psi\le2^{nL}\sigma_n^\psi$.

The adaptive chain rule gives the threshold $d$ once $d_\alpha$
converges to $d$ as $\alpha\downarrow1$. The next theorem establishes
this convergence.

\begin{theorem}[Right continuity of the regularized divergence]
\label{thm:renyi-continuity}
For channels $\N,\M:A\to B$ with $d=D^\infty(\N\Vert\M)<\infty$,
\begin{equation}
 \lim_{\alpha\downarrow1}\widetilde D_\alpha^\infty(\N\Vert\M)=d.
 \label{eq:renyi-continuity}
\end{equation}
\end{theorem}

In fact, let $L$ be the domination rate from \cref{lem:finite}.
For $d<S<L$, take an approximation with error at most $Ke^{-\gamma n}$
from \cref{thm:approximation}. We will prove that, for $1<\alpha\le2$,
\begin{equation}
 \widetilde D_\alpha^\infty(\N\Vert\M)
 \le\max\left\{S,\ L-\frac{2\gamma}{(\alpha-1)\ln2}\right\}.
 \label{eq:renyi-quantitative}
\end{equation}
To prove this bound, we must control the residual after weighting by a
negative power of $\sigma$. Small trace alone is insufficient; the
following lemma also uses domination by $\sigma$, without requiring
a lower bound on its nonzero eigenvalues.

\begin{lemma}[Small-trace dominated terms]\label{lem:renyi-moment}
Let $\sigma\in\mathcal S(\mathsf H)$, $t\ge0$, and $1<\alpha\le2$.
Every positive operator $\tau\le t\sigma$ satisfies
$\widetilde Q_\alpha(\tau\Vert\sigma)\le t^{\alpha-1}\Tr\tau$.
Consequently, if $F_0,F_1:\mathsf E\to\mathsf H$ are linear maps,
$F=F_0+F_1$, and $F_jF_j^\dagger\le t_j\sigma$ with $t_j\ge0$,
then
\begin{equation}
 \widetilde Q_\alpha(FF^\dagger\Vert\sigma)^{1/{2\alpha}}
 \le\sum_{j=0}^1t_j^{(\alpha-1)/(2\alpha)}\|F_j\|_2^{1/\alpha}.
 \label{eq:renyi-factor-sum}
\end{equation}
\end{lemma}

\begin{proof}
The case $t=0$ has $\tau=0$. Otherwise work on $\supp\sigma$ and
put $B=\sigma^{-(\alpha-1)/{2\alpha}}\tau
\sigma^{-(\alpha-1)/{2\alpha}}$. Domination gives
$B\le t\sigma^{1/\alpha}$. The function $x\mapsto x^{\alpha-1}$
is operator monotone for $1<\alpha\le2$, so
$B^{\alpha-1}\le t^{\alpha-1}\sigma^{(\alpha-1)/\alpha}$.
Multiplying by $B$ inside the trace preserves this inequality:
\begin{equation}
 \Tr B^\alpha
 \le t^{\alpha-1}\Tr B\sigma^{(\alpha-1)/\alpha}
 =t^{\alpha-1}\Tr\tau.
 \label{eq:renyi-small-trace}
\end{equation}
For the second assertion,
$\widetilde Q_\alpha(FF^\dagger\Vert\sigma)^{1/{2\alpha}}
=\|\sigma^{-(\alpha-1)/{2\alpha}}F\|_{2\alpha}$.
Apply the Schatten-norm triangle inequality to $F_0+F_1$ and use
\eqref{eq:renyi-small-trace} with $\tau=F_jF_j^\dagger$.
As $\Tr F_jF_j^\dagger=\|F_j\|_2^2$, this gives
\eqref{eq:renyi-factor-sum}, including zero terms.
\end{proof}

\begin{proof}[Proof of \cref{thm:renyi-continuity}]
If $d=L$, the bounds $d\le d_\alpha\le L$ already prove continuity.
Otherwise fix $d<S<L$, the dilations, and auxiliary maps $C_n$ from
\cref{thm:approximation}. \cref{cor:exact} provides
$C_*$ with $\op{C_*}\le2^{L/2}$ in the same environments.

Fix any pure input with a reference for $n$ uses, and let
$\rho,\sigma$ be its output states. Reshape the null output
purification, its approximation, and their difference into matrices
$F,F_0,F_1$ whose columns are indexed by the null environment. Then
$F=F_0+F_1$, $FF^\dagger=\rho$, $\|F_0\|_2\le1+e_n$, and
$\|F_1\|_2\le e_n$, where $e_n\le Ke^{-\gamma n}$.
If $G$ is the alternative purification's column matrix, then
$GG^\dagger=\sigma$, $F_0=GC_n^{\mathsf T}$, and
$F_1=G(C_*^{\otimes n}-C_n)^{\mathsf T}$.
The last identity follows from the exact environmental representation.
Consequently,
\begin{equation*}
 F_0F_0^\dagger\le2^{nS}\sigma,\qquad
 F_1F_1^\dagger\le4\,2^{nL}\sigma,
\end{equation*}
since $\op{C_*^{\otimes n}-C_n}\le2 \cdot 2^{nL/2}$.
Here $F=F_0+F_1$ is a decomposition of purification matrices, not a
positive-operator decomposition of $\rho$; the cross terms are handled
by the Schatten-norm triangle inequality.

\cref{lem:renyi-moment} gives, uniformly over the input,
\begin{align*}
 \widetilde Q_\alpha(\rho\Vert\sigma)^{1/{2\alpha}}
 &\le(1+e_n)^{1/\alpha}2^{nS(\alpha-1)/{2\alpha}}\\
 &\quad+e_n^{1/\alpha}(4\,2^{nL})^{(\alpha-1)/{2\alpha}}.
\end{align*}
For fixed $S$ and $\alpha$, the two summands are bounded by constant
multiples of $2^{nu_0}$ and $2^{nu_1}$, where
$u_0=(\alpha-1)S/{2\alpha}$ and
$u_1=(\alpha-1)L/{2\alpha}-\gamma/(\alpha\ln2)$.
The constants are independent of the input and $n$. Hence the logarithm
of their sum is at most $n\max\{u_0,u_1\}+O(1)$, uniformly over all
inputs. Multiplying by $2\alpha/(\alpha-1)$, optimizing over the
input, and dividing by $n$ proves \eqref{eq:renyi-quantitative} upon
taking the limit.

For each fixed $S\in(d,L)$, the second term of the maximum in
\eqref{eq:renyi-quantitative} is below $S$ whenever $\alpha>1$ is
sufficiently close to one. Therefore
$d\le\liminf_{\alpha\downarrow1}d_\alpha
\le\limsup_{\alpha\downarrow1}d_\alpha\le S$.
Letting $S\downarrow d$ completes the proof.
\end{proof}

The exponential error is needed here because it yields the term
$-2\gamma/((\alpha-1)\ln2)$ in \eqref{eq:renyi-quantitative}.
A subexponential error bound would give no such reduction after division
by $n$.

\subsection{The adaptive converse}\label{sec:adaptive}

The Fawzi--Fawzi chain rule~\cite[Corollary~5.2]{FawziFawzi} applies
to different input states. For $\alpha>1$, channels $\N,\M:A\to B$,
and states $\rho,\sigma$ on a common reference and $A$, it states
\begin{equation}
 \begin{split}
 &\widetilde D_\alpha\bigl((\id\otimes\N)(\rho)
                    \Vert(\id\otimes\M)(\sigma)\bigr)\\
 &\qquad\le\widetilde D_\alpha(\rho\Vert\sigma)
           +\widetilde D_\alpha^\infty(\N\Vert\M).
 \end{split}
 \label{eq:renyi-chain}
\end{equation}
Combined with continuity, this bounds every use in an adaptive
protocol at a rate arbitrarily close to $d$.

\begin{proof}[Proof of \cref{thm:adaptive}]
Assume $d<\infty$ and fix $R>d$. By
\cref{thm:renyi-continuity}, choose $1<\alpha\le2$ such that
$d_\alpha=\widetilde D_\alpha^\infty(\N\Vert\M)<R$.
The two hypotheses start from the same state, with divergence zero.
Inductively, \eqref{eq:renyi-chain} adds at most $d_\alpha$ at each
use, while data processing prevents any increase under the
intervening CPTP maps. The final-state divergence is therefore at most
$nd_\alpha$. Each step permits an arbitrary common reference, so the
bound is independent of the memory dimensions.

Consider a final test with $q\le2^{-nR}$.
If $p=0$, the desired bound is immediate.
If $p>0$, finite final-state divergence implies $q>0$, and $R>d\ge0$ gives
$q<1$.
Data processing under the binary measurement, keeping only the term
$p^\alpha q^{1-\alpha}$ inside the logarithm, gives
\begin{align*}
 nd_\alpha
 &\ge\frac1{\alpha-1}
 \log\bigl[p^\alpha q^{1-\alpha}
           +(1-p)^\alpha(1-q)^{1-\alpha}\bigr]\\
 &\ge\frac\alpha{\alpha-1}\log p-\log q.
\end{align*}
Rearranging and using $-\log q\ge nR$ proves
\begin{equation*}
 p\le2^{-n\frac{\alpha-1}{\alpha}(R-d_\alpha)}.
\end{equation*}
Thus \eqref{eq:adaptive-power} holds for every $n\ge1$ with
$c_R=(\alpha-1)(R-d_\alpha)/\alpha>0$.

For fixed $\varepsilon$, the requirement $p\ge1-\varepsilon$ is
incompatible with this bound for sufficiently large $n$. Hence the
adaptive limit superior in \eqref{eq:adaptive-stein} is at most $d$.
The parallel block strategies in \cref{prop:parallel}
are adaptive strategies as well and give the matching lower bound.
If $d=+\infty$, \cref{lem:finite} and
\cref{prop:infinite} provide parallel tests with zero
type-II error at all sufficiently large blocklengths. Both fixed-error
limits are then $+\infty$.
\end{proof}

This proves positivity of the strong-converse exponent characterized
in~\cite[Theorem~5.5]{FawziFawzi} at every $R>d$.

\section{Uniform Exponential Approximation}\label{sec:approximation}

We now prove \cref{thm:approximation}. The one-shot comparison
below applies to arbitrary channel pairs. Finite divergence is needed
only afterwards, to obtain both a constant-error approximation from
relative entropy and an exact representation at finite cost rate.
Two tensor constructions then improve the approximation.

\subsection{From testing to a uniform approximation}\label{sec:one-shot}

We use the state hockey-stick quantity
$E_t(\rho\Vert\sigma)=\Tr(\rho-t\sigma)_+$ for $t\ge0$,
extending the definition for $t\ge1$ in~\cite{SharmaWarsi}.
For channels $\N,\M:A\to B$, define its channel extension by
\begin{equation}
 E_t(\N\Vert\M)
 =\sup_{\ket\psi}\Tr(\rho_1^\psi-t\sigma_1^\psi)_+
 =\sup_{\ket\psi,\,0\le T\le I_{RB}}(p-tq),
 \label{eq:hockey}
\end{equation}
where $\rho_1^\psi,\sigma_1^\psi$ are the output states in
\eqref{eq:outputs}, and $p=\Tr T\rho_1^\psi$ and
$q=\Tr T\sigma_1^\psi$. The second equality follows from the
variational formula for the positive part. Pure inputs with
$R\simeq A$ suffice, and $0\le E_t\le1$.
For $t\ge1$, this is the unrestricted-measurement case
of~\cite[Def.~3]{NuradhaSinghWilde}. For $0\le t<1$, that reference
subtracts $(1-t)_+$.

The next semidefinite representation converts a bound on all inputs
and tests into a single positive operator. It is an instance of the
semidefinite approach to completely bounded norms~\cite{WatrousSDP};
we give the primal and dual explicitly to fix the normalization.

\begin{lemma}[Channel testing duality]\label{lem:sdp}
Let $\N,\M:A\to B$ be channels, let $t\ge0$, and set
$\Delta=J_\N-tJ_\M$. Then
\begin{align}
 E_t(\N\Vert\M)
 &=\max_{\substack{\omega\in\mathcal S(A)\\
                    0\le Q\le\omega\otimes I_B}}\Tr\Delta Q
       \label{eq:sdp-primal}\\
 &=\min_{Y\ge0,\ Y\ge\Delta}\op{\Tr_BY}.
       \label{eq:sdp-dual}
\end{align}
Here $Q,Y$ act on $A\otimes B$. Both extrema are attained.
\end{lemma}

\begin{proof}
Identify the reference with $A$. Up to an isometry on the reference,
a pure input has the form $(\sqrt\omega\otimes I_A)\ket\Omega$
for some $\omega\in\mathcal S(A)$. Pulling the output effect $T$
through $\sqrt\omega$ gives
$Q=(\sqrt\omega\otimes I_B)T(\sqrt\omega\otimes I_B)$,
which satisfies $0\le Q\le\omega\otimes I_B$ and
$p-tq=\Tr\Delta Q$.
Conversely, this order constraint forces $Q$ to be supported on
$\supp\omega\otimes B$. Congruence by the inverse of
$\sqrt\omega$ on its support gives an effect on that subspace;
extending it by zero gives a valid $T$. This proves
\eqref{eq:sdp-primal}.

Introduce a positive multiplier $Y$ for
$Q\le\omega\otimes I_B$ and a real multiplier $\lambda$ for
$\Tr\omega=1$. On $Q,\omega\ge0$, the Lagrangian is
\[
 \lambda+\Tr[(\Delta-Y)Q]
       +\Tr[(\Tr_BY-\lambda I_A)\omega].
\]
Its supremum is finite exactly when $Y\ge\Delta$ and
$\Tr_BY\le\lambda I_A$. For positive $Y$, the least such
$\lambda$ is $\op{\Tr_BY}$, giving the dual in
\eqref{eq:sdp-dual}. Strict primal feasibility is witnessed by
$\omega=I_A/\dim A$ and $Q=\eta I_{AB}$ with
$0<\eta<1/\dim A$. Strict dual feasibility follows by taking
$Y=yI_{AB}$ with $y>\max\{0,\op\Delta\}$ and then
$\lambda>y\dim B$. Semidefinite strong duality applies.
The primal feasible set is compact. On a dual sublevel set,
$Y\ge0$ and $\Tr Y\le(\dim A)\op{\Tr_BY}$ bound the norm
of $Y$; the set is also closed. This proves attainment on both sides.
\end{proof}

Fix dilations $V_N:A\to B\otimes\mathsf E_N$ and
$V_M:A\to B\otimes\mathsf E_M$. For $t\ge0$, let
\begin{equation}
 s_t=\min_{\substack{C:\mathsf E_M\to\mathsf E_N\\
                       \op C\le\sqrt t}}
       \op{V_N-(I_B\otimes C)V_M}.
 \label{eq:st-distance}
\end{equation}
The following comparison turns testing control into an approximation
in the prescribed environment. This fixed codomain will make it
possible to subtract and tensor auxiliary maps constructed at different
rates.

\begin{proposition}
\label{prop:environment}
Let $\N,\M:A\to B$ be channels. Fix Stinespring isometries
$V_N:A\to B\otimes\mathsf E_N$ and
$V_M:A\to B\otimes\mathsf E_M$ for $\N$ and $\M$, respectively,
and define $s_t$ by \eqref{eq:st-distance} using these isometries.
Then
\begin{equation*}
 s_t^2\le E_t(\N\Vert\M)\le2s_t-s_t^2
 \qquad(t\ge0).
\end{equation*}
\end{proposition}

\begin{proof}
The set of maps $C:\mathsf E_M\to\mathsf E_N$ satisfying
$\op C\le\sqrt t$ is compact in finite dimensions. Since
$C\mapsto\op{V_N-(I_B\otimes C)V_M}$ is continuous, the minimum
defining $s_t$ is attained. The choice
$C=0$ gives $s_t\le1$. Take a dual optimizer $Y$ from
\cref{lem:sdp}, and put $A_0=tJ_\M$ and $H=A_0+Y$.
The dual constraint $Y\ge J_\N-tJ_\M$ implies
$J_\N\le tJ_\M+Y=H$. Write the fixed dilations as
$V_N\ket x=\sum_j N_j\ket x\otimes\ket j$ and
$V_M\ket x=\sum_k M_k\ket x\otimes\ket k$ in orthonormal environment bases.
As in \cref{sec:algebra}, collect the vectorized Kraus operators into
the column matrices
$F_N=[\vecop(N_1)\ \cdots\ \vecop(N_u)]$ and
$F_M=[\vecop(M_1)\ \cdots\ \vecop(M_v)]$, where
$\vecop(K)=\sum_i\ket i\otimes K\ket i$.
These satisfy $F_NF_N^\dagger=J_\N$ and
$F_MF_M^\dagger=J_\M$. Define
$F_0=A_0H^+F_N$ and $F_1=YH^+F_N$.
Here $H^+$ is the Moore--Penrose pseudoinverse of $H$.
Since $\ran F_N\subseteq\supp H$, we have
$F_0+F_1=HH^+F_N=F_N$. Moreover,
\begin{equation}
 F_0F_0^\dagger\le A_0H^+A_0\le A_0,
 \qquad F_1F_1^\dagger\le YH^+Y\le Y.
 \label{eq:factor-order}
\end{equation}
The first inequality in each chain uses $J_\N\le H$. For the second,
let $0\le Z\le H$ and put $K=(H^+)^{1/2}Z^{1/2}$. Then
$KK^\dagger\le I$, hence $K^\dagger K\le I$.
Congruence of $Z^{1/2}H^+Z^{1/2}\le I$ by $Z^{1/2}$ gives
$ZH^+Z\le Z$. Apply this with $Z=A_0$ and $Z=Y$.

Since $F_0F_0^\dagger\le A_0=tF_MF_M^\dagger$,
\cref{lem:factor} gives $F_0=F_MD$ for a map
$D:\mathsf E_N\to\mathsf E_M$ with $\op D\le\sqrt t$.
Set $C=D^{\mathsf T}:\mathsf E_M\to\mathsf E_N$.
Transposition preserves the operator norm, so $C$ is feasible in
\eqref{eq:st-distance}. Under the column convention of
\cref{sec:algebra}, applying $I_B\otimes C$ to $V_M$ changes its
Kraus-column matrix to $F_MC^{\mathsf T}=F_MD=F_0$.
More explicitly, write $C\ket k=\sum_j C_{jk}\ket j$.
For any input vector $\ket x$, the residual acts as
\begin{equation*}
 \bigl[V_N-(I_B\otimes C)V_M\bigr]\ket x
 =\sum_j\left(N_j-\sum_k C_{jk}M_k\right)\ket x\otimes\ket j.
\end{equation*}
Its $j$th column is therefore
$\vecop(N_j-\sum_k C_{jk}M_k)
=\vecop(N_j)-\sum_k C_{jk}\vecop(M_k)$.
This is precisely the $j$th column of $F_N-F_0$, so the residual
has column matrix $F_N-F_0=F_1$.

The reshaping identity \eqref{eq:reshape-norm} expresses the squared
operator norm of this residual as $\op{\Tr_BF_1F_1^\dagger}$.
By \eqref{eq:factor-order}, $F_1F_1^\dagger\le Y$; taking the
partial trace preserves this order, as does taking the operator norm
of positive operators. Since $Y$ is a dual optimizer, we obtain
\begin{equation*}
 \op{V_N-(I_B\otimes C)V_M}^2
 =\op{\Tr_BF_1F_1^\dagger}
 \le\op{\Tr_BY}=E_t(\N\Vert\M).
\end{equation*}
Minimizing over feasible $C$ gives $s_t^2\le E_t(\N\Vert\M)$.
The argument also applies at $t=0$, when the factorization forces
$D=0$ and hence $C=0$.

For the upper bound, choose an auxiliary map $C$ attaining the minimum
in \eqref{eq:st-distance}. Its approximation error is $s_t$ and
$\op C\le\sqrt t$, so \cref{lem:amplitude} gives
\begin{equation*}
 \sqrt p\le\op C\sqrt q+s_t\le\sqrt{tq}+s_t
\end{equation*}
for every input and effect, with acceptance probabilities $p,q$.
If the testing objective $p-tq$ is positive, then
$z=\sqrt p-\sqrt{tq}>0$, and the preceding inequality gives
$z\le s_t\le1$. Factoring the difference of squares yields
\begin{equation*}
 p-tq=z(\sqrt p+\sqrt{tq})
      =z(2\sqrt p-z)
      \le2z-z^2\le2s_t-s_t^2.
\end{equation*}
Here the first inequality uses $\sqrt p\le1$, and the second uses
that $z\mapsto2z-z^2$ is increasing on $[0,1]$.
If $p-tq\le0$, the same bound holds because $2s_t-s_t^2\ge0$.
Thus every input and effect satisfies $p-tq\le2s_t-s_t^2$.
Taking the supremum of $p-tq$ over all inputs and effects gives
$E_t(\N\Vert\M)$ by \eqref{eq:hockey}, proving the upper bound.
\end{proof}

Uniform CP approximations from testing information also appear
in~\cite{Gour}. Here the approximation is expressed by one linear
auxiliary map with a fixed codomain, independently of both the input
and the test. This is the form required by the tensor constructions.
For a channel pair with $d=D^\infty(\N\Vert\M)<\infty$, abbreviate
$E_n(r)=E_{2^{nr}}(\N^{\otimes n}\Vert\M^{\otimes n})$.
The next lemma supplies the only information-theoretic estimate
needed to begin amplification.

\begin{lemma}\label{lem:weak-tail}
Let $\N,\M:A\to B$ be channels with $d=D^\infty(\N\Vert\M)<\infty$.
For every $r>0$ and $n\ge1$,
\begin{equation}
 E_n(r)\le\min\left\{1,\frac{nd+1}{nr}\right\}.
 \label{eq:weak-tail}
\end{equation}
\end{lemma}

\begin{proof}
Fix a pure input $\ket\psi_{RA^n}$ for $n$ parallel channel uses,
and let $\rho_n^\psi,\sigma_n^\psi$ be the output states defined in
\eqref{eq:outputs}. Choose an effect $0\le T\le I_{RB^n}$
corresponding to acceptance of the null hypothesis. The acceptance
probabilities are $p=\Tr T\rho_n^\psi$ under $\N^{\otimes n}$
and $q=\Tr T\sigma_n^\psi$ under $\M^{\otimes n}$.
By \eqref{eq:hockey}, $E_n(r)$ is the supremum of $p-2^{nr}q$
over these inputs and effects. Since the claimed upper bound is
nonnegative, every choice with $p-2^{nr}q\le0$ already satisfies it.
It therefore suffices to consider choices with $p-2^{nr}q>0$.
For such a choice, $p>0$ and $q<2^{-nr}$, since $p\le1$.
By \cref{lem:finite}, there is a
finite $c\ge1$ such that $\rho_n^\psi\le c^n\sigma_n^\psi$.
Taking the trace against the positive effect $T$ gives
$p=\Tr T\rho_n^\psi\le c^n\Tr T\sigma_n^\psi=c^nq$.
Thus $q=0$ would force $p=0$, contradicting $p>0$.
By the definition of the channel relative entropy and
\eqref{eq:d-setup}, this input satisfies
$D(\rho_n^\psi\Vert\sigma_n^\psi)\le a_n\le nd$.
The binary measurement $\{T,I-T\}$ maps the two output states to
the probability distributions $(p,1-p)$ and $(q,1-q)$.
Relative entropy cannot increase under this measurement, so
\begin{align*}
 nd\ge D(\rho_n^\psi\Vert\sigma_n^\psi)
 &\ge p\log\frac pq+(1-p)\log\frac{1-p}{1-q}\\
 &=p\log(1/q)+(1-p)\log(1/(1-q))-\binary(p)\\
 &\ge p\log(1/q)-\binary(p).
\end{align*}
The equality uses
$\binary(p)=-p\log p-(1-p)\log(1-p)$, and the last inequality
drops the nonnegative term $(1-p)\log(1/(1-q))$.
Since $0<q<2^{-nr}$, we have $\log(1/q)>nr$; also
$\binary(p)\le1$. Consequently, $nd\ge pnr-1$, or
$p\le(nd+1)/(nr)$. Since $p-2^{nr}q\le p$, every positive
objective value is at most $(nd+1)/(nr)$. The zero effect and
$E_n(r)\le1$ prove \eqref{eq:weak-tail}.

\end{proof}

We now convert this testing bound into the approximation needed for
the tensor argument. Fix Stinespring isometries $V_N,V_M$ for the
two channels, and choose a rate $r>d$. Since $d/r<1$, we can choose
a constant $\delta$ with $\sqrt{d/r}<\delta<1$. By
\cref{lem:weak-tail},
\begin{equation*}
 E_n(r)\le\frac{nd+1}{nr}
 =\frac dr+\frac1{nr}\le\delta^2
\end{equation*}
for all sufficiently large $n$; the last inequality holds because
$\delta^2-d/r>0$ is fixed.

Apply \cref{prop:environment} to the channels
$\N^{\otimes n},\M^{\otimes n}$, using the fixed tensor-power
dilations $V_N^{\otimes n},V_M^{\otimes n}$ and $t=2^{nr}$.
The bound $s_t^2\le E_t$ says that the minimum squared approximation
error under the norm constraint $\op C\le\sqrt t$ is at
most $E_n(r)$. Since this minimum is attained, there is a map
$C_n^{\mathrm w}:\mathsf E_M^{\otimes n}\to\mathsf E_N^{\otimes n}$
such that
\begin{equation*}
 \op{C_n^{\mathrm w}}^2\le2^{nr},\qquad
 \op{V_N^{\otimes n}-(I_{B^n}\otimes C_n^{\mathrm w})V_M^{\otimes n}}
 \le\sqrt{E_n(r)}\le\delta.
\end{equation*}
Thus the weak testing bound supplies an auxiliary map at
any rate $r>d$, with an operator-norm error uniformly bounded by a
constant strictly below one. The same auxiliary map works for every
input and reference system. This estimate alone does not show that
the error tends to zero as $n$ grows.

When $r<L$, the next lemma combines these auxiliary maps with the exact auxiliary map
$C_*$ from \cref{cor:exact}, whose tensor powers have zero error at
the finite rate $L$. The weak auxiliary maps supply the lower rate $r$,
while the exact auxiliary map supplies an accurate starting point and
a finite bound on the cost of corrections. The tensor argument uses
the strict inequality $\delta<1$ to control terms containing many
corrections, and obtains vanishing error at every rate above $r$.
If $r\ge L$, the exact auxiliary maps already meet the rate bound.

\subsection{Two-rate tensor amplification}\label{sec:amplification}

We improve the approximation in two stages. In this subsection, we
combine a low-rate approximation of constant error with an accurate
representation at a higher rate, obtaining vanishing error at every
rate above the low rate. In \cref{sec:exponential}, we then fix one
sufficiently accurate block and expand it together with its exact
residual to obtain exponentially small error, with an arbitrarily
small further increase in rate.

The argument is algebraic and does not use channel divergences.
For isometries $V:A\to B\otimes\mathsf E$ and $W:A\to B\otimes\mathsf F$, write
$\mathcal T_n(C)=(I_{B^n}\otimes C)W^{\otimes n}$ for
$C:\mathsf F^{\otimes n}\to\mathsf E^{\otimes n}$.
Canonical permutations group the output and environment factors.
The maps $\mathcal T_n$ are linear and satisfy
\begin{equation}
 \mathcal T_{n+m}(C\otimes D)
 =\mathcal T_n(C)\otimes\mathcal T_m(D).
 \label{eq:T-tensor}
\end{equation}
Throughout this and the next subsection, an auxiliary map at blocklength
$n$ is a linear map from $\mathsf F^{\otimes n}$ to
$\mathsf E^{\otimes n}$. Its cost rate is at most $r$ when its squared
norm is at most $2^{nr}$. The next lemma combines a low-rate
constant-error approximation with a finite-rate exact representation.

\begin{lemma}[Two-rate tensor amplification]\label{lem:amplification}
Suppose $V=\mathcal T_1(C_*)$ with $\op{C_*}\le2^{L/2}$ for
some $L\ge0$. Assume that, for some $0\le r<L$ and $0\le\delta<1$,
for every sufficiently large $n$, there exists an auxiliary map $C_n^{\mathrm w}$
with $\op{C_n^{\mathrm w}}\le2^{nr/2}$ and
$\op{\mathcal T_n(C_n^{\mathrm w})-V^{\otimes n}}\le\delta$.

Then for every $R>r$, there exist $K,\gamma>0$ and $a\in(0,1]$ such that,
for all sufficiently large $N$, there are auxiliary maps $C_N$ with
\begin{equation}
 \op{C_N}\le2^{NR/2},\qquad
 \op{\mathcal T_N(C_N)-V^{\otimes N}}\le Ke^{-\gamma N^a}.
 \label{eq:amp-conclusion}
\end{equation}
\end{lemma}

\begin{proof}
For a rate $S$, let $\mathsf P(S)$ denote the following property:
there exist $K,\gamma>0$ and $a\in(0,1]$ such that, for every sufficiently
large $n$, there exists an auxiliary map $C_n^{\mathrm s}$ with
$\op{C_n^{\mathrm s}}\le2^{nS/2}$ and error
$e_n=\op{\mathcal T_n(C_n^{\mathrm s})-V^{\otimes n}}
\le Ke^{-\gamma n^a}$. All constants and the initial blocklength
may depend on $S$, but not on $n$. The exact auxiliary maps
$C_*^{\otimes n}$ give $\mathsf P(L)$ with zero error and $a=1$.

\smallskip
\noindent\emph{Step 1: Fix the truncation fraction.}
Set $A_0=1+\delta$ and $B_0=(1+\delta)/2<1$.
At $x=1$, the continuous function
$(\ln2)\binary(x)+(1-x)\ln A_0+x\ln B_0$ equals $\ln B_0<0$.
There are therefore $\beta\in(0,1)$ and $\kappa>0$ such that
\begin{equation}
 (\ln2)\binary(x)+(1-x)\ln A_0+x\ln B_0\le-\kappa
 \quad\text{for }\beta\le x\le1.
 \label{eq:entropy-tail}
\end{equation}
Keep these constants fixed throughout the iteration. Choosing
$\beta$ close to one ensures that the decay of $B_0^j$ outweighs
the number of omitted terms in the tensor expansion.

\smallskip
\noindent\emph{Step 2: Lower the rate of an accurate approximation.}
We show that
\begin{equation}
 \mathsf P(S)\Longrightarrow\mathsf P(S')
 \quad\text{if }S>r\text{ and }S'>(1-\beta)r+\beta S.
 \label{eq:bootstrap}
\end{equation}
Assume $\mathsf P(S)$. At a sufficiently large blocklength $n$,
put $X=\mathcal T_n(C_n^{\mathrm w})$,
$Y=\mathcal T_n(C_n^{\mathrm s})$, and $H=Y-X$.
Then $\op X\le A_0$. Since $e_n\to0$, we may also assume
$e_n\le(1-\delta)/2$, giving $\op H\le\delta+e_n\le B_0$.
Linearity expresses $H$ by the auxiliary map
$C_n^{\mathrm s}-C_n^{\mathrm w}$, whose norm is at most
$2\cdot2^{nS/2}$ because $S>r$.

Expand $Y^{\otimes m}=(X+H)^{\otimes m}$, and keep only terms with at most
$\lfloor\beta m\rfloor$ factors $H$.
In detail, for $J\subseteq\{1,\ldots,m\}$ let $G_{J,j}=H$ if $j\in J$ and
$G_{J,j}=X$ otherwise, and set
\begin{equation}
 Z_{n,m}=\sum_{\substack{J\subseteq\{1,\ldots,m\}\\
                         |J|\le\lfloor\beta m\rfloor}}
              \bigotimes_{j=1}^mG_{J,j}.
 \label{eq:truncation}
\end{equation}
All summands have the same domain and codomain. The expansion selects
$X$ or $H$ independently on each tensor factor.
Using $\binom mj\le2^{m\binary(j/m)}$ and
\eqref{eq:entropy-tail}, the omitted terms obey
\begin{align*}
 \op{Z_{n,m}-Y^{\otimes m}}
 &\le\sum_{j>\beta m}\binom mj A_0^{m-j}B_0^j\\
 &\le(m+1)e^{-\kappa m}.
\end{align*}

To represent $Z_{n,m}$ by an auxiliary map, recall that
$X=\mathcal T_n(C_n^{\mathrm w})$ and
$H=\mathcal T_n(C_n^{\mathrm s}-C_n^{\mathrm w})$.
For each subset $J\subseteq\{1,\ldots,m\}$, define
\begin{equation*}
 A_{J,j}=\begin{cases}
 C_n^{\mathrm s}-C_n^{\mathrm w},&j\in J,\\
 C_n^{\mathrm w},&j\notin J,
 \end{cases}
 \qquad
 D_{n,m}=\sum_{\substack{J\subseteq\{1,\ldots,m\}\\
                         |J|\le\lfloor\beta m\rfloor}}
                  \bigotimes_{j=1}^m A_{J,j}.
\end{equation*}
Each tensor product is a map from $\mathsf F^{\otimes nm}$ to
$\mathsf E^{\otimes nm}$. By \eqref{eq:T-tensor}, applying
$\mathcal T_{nm}$ to this tensor product gives
$\bigotimes_{j=1}^m G_{J,j}$, the corresponding summand in
\eqref{eq:truncation}. Linearity then gives
$\mathcal T_{nm}(D_{n,m})=Z_{n,m}$.
A subset of size $j$ contributes $m-j$ factors with norm at most
$2^{nr/2}$ and $j$ factors with norm at most $2\cdot2^{nS/2}$.
Multiplying these bounds gives $2^j2^{n[(m-j)r+jS]/2}$ for each
tensor product. There are $\binom mj$ subsets of size $j$, so
the triangle inequality gives
\begin{align}
 \op{D_{n,m}}
 &\le\sum_{j=0}^{\lfloor\beta m\rfloor}\binom mj2^j
          2^{n[(m-j)r+jS]/2}\notag\\
 &\le2^{nm[(1-\beta)r+\beta S]/2}
       \sum_{j=0}^{\lfloor\beta m\rfloor}\binom mj2^j\notag\\
 &\le3^m2^{nm[(1-\beta)r+\beta S]/2}.
 \label{eq:auxiliary-map-cost}
\end{align}
The second inequality uses $S>r$ and $j\le\beta m$ to bound
$(m-j)r+jS\le m[(1-\beta)r+\beta S]$. The last inequality extends
the sum to $j=m$ and uses $\sum_{j=0}^m\binom mj2^j=3^m$.

The truncation estimate above bounds $\op{Z_{n,m}-Y^{\otimes m}}$.
To bound the error relative to the target $V^{\otimes nm}$, we must
also control $\op{Y^{\otimes m}-(V^{\otimes n})^{\otimes m}}$,
since $Y$ only approximates $V^{\otimes n}$ with error $e_n$.
The telescoping expansion of this difference
has $m$ terms, each with one factor $Y-V^{\otimes n}$ and all
other factors of norm at most $1+e_n$. Hence the total error is
\begin{equation}
 \op{\mathcal T_{nm}(D_{n,m})-V^{\otimes nm}}
 \le(m+1)e^{-\kappa m}+me_n(1+e_n)^{m-1}.
 \label{eq:block-error}
\end{equation}

\smallskip
\noindent\emph{Step 3: Cover every sufficiently large blocklength.}
For a large integer $N$, choose $n=\lfloor\sqrt N\rfloor$,
$m=\lfloor N/n\rfloor$, and $\ell=N-nm$, so $0\le\ell<n$.
Pad with the exact auxiliary map:
$C_N=D_{n,m}\otimes C_*^{\otimes\ell}$.
The associated operator is $Z_{n,m}\otimes V^{\otimes\ell}$;
since $V$ is an isometry, padding does not increase the error.
The empty tensor power is interpreted as the scalar identity.
Writing $\bar S=(1-\beta)r+\beta S$, \eqref{eq:auxiliary-map-cost}
gives
\begin{equation*}
 \op{C_N}^2\le3^{2m}2^{nm\bar S+\ell L},\qquad
 \frac{nm\bar S+\ell L+2m\log3}{N}=\bar S+o(1).
\end{equation*}
The last identity uses $nm/N\to1$, $\ell/N\to0$, and $m/N\to0$.
Thus $\op{C_N}\le2^{NS'/2}$ for every fixed $S'>\bar S$ and
all sufficiently large $N$.

We bound the two error terms in \eqref{eq:block-error} separately.
Since $n=\lfloor\sqrt N\rfloor$ and $m=\lfloor N/n\rfloor$,
both are bounded above and below by positive constant multiples of
$\sqrt N$ for all large $N$. In particular, $n^a\ge bN^{a/2}$
for some $b>0$, and
\begin{equation*}
 me_n\le mK e^{-\gamma n^a}
 \le K_1\sqrt N\,e^{-\gamma bN^{a/2}}\longrightarrow0.
\end{equation*}
Thus $(1+e_n)^{m-1}\le e^{me_n}$ is eventually at most $2$.
The second error term therefore satisfies
\begin{equation*}
 me_n(1+e_n)^{m-1}
 \le2K_1\sqrt N\,e^{-\gamma bN^{a/2}}
 \le K_2e^{-(\gamma b/2)N^{a/2}}
\end{equation*}
for all large $N$. The last inequality uses
$\log N=o(N^{a/2})$, so the factor $\sqrt N$ can be absorbed
by halving the decay constant.

For the first error term, $m$ is of order $\sqrt N$, so
$(m+1)e^{-\kappa m}\le K_3e^{-\kappa_1\sqrt N}$ for suitable
$K_3,\kappa_1>0$, by the same argument for the polynomial prefactor.
Since $a\le1$, we have $\sqrt N\ge N^{a/2}$, giving the bound
$K_3e^{-\kappa_1N^{a/2}}$.
Adding the two bounds and recalling that exact padding preserves the
error yields
\begin{equation*}
 \op{\mathcal T_N(C_N)-V^{\otimes N}}
 \le K'e^{-\gamma' N^{a/2}}
\end{equation*}
for suitable $K',\gamma'>0$. Hence the new approximation satisfies
$\mathsf P(S')$ with exponent $a/2$, proving \eqref{eq:bootstrap}.

\smallskip
\noindent\emph{Step 4: Iterate a fixed number of times.}
Let $\theta=(1+\beta)/2\in(\beta,1)$ and
$S_j=r+\theta^j(L-r)$ for $j=0,1,\ldots$.
These rates start at $S_0=L$ and decrease to $r$. Moreover,
\begin{equation*}
 S_{j+1}-\bigl[(1-\beta)r+\beta S_j\bigr]
 =(\theta-\beta)\theta^j(L-r)>0,
\end{equation*}
so each successive pair meets the strict rate condition in
\eqref{eq:bootstrap}. Starting from the exact representation, which
gives $\mathsf P(L)$ with exponent $a=1$, we may therefore apply
that implication repeatedly. Each application halves the error
exponent, so after $j$ iterations we have $\mathsf P(S_j)$ with
$a=2^{-j}$.

Now fix a target rate $r<R<L$. Since $\theta^j(L-r)\to0$, a
finite $j$ satisfies $\theta^j(L-r)<R-r$, or equivalently $S_j<R$.
The auxiliary maps supplied by $\mathsf P(S_j)$ then satisfy
$\op{C_N}^2\le2^{NS_j}\le2^{NR}$ and have error at most
$Ke^{-\gamma N^{2^{-j}}}$ for all sufficiently large $N$.
We choose this $j$ before letting $N$ grow. Thus $2^{-j}>0$ is a
fixed exponent, and the constants and initial blocklength may depend
on $R$ but not on $N$. If $R\ge L$, the exact auxiliary maps
already satisfy the rate bound with zero error.
This proves \eqref{eq:amp-conclusion}.
\end{proof}

\subsection{Exponential accuracy from a fixed block}\label{sec:exponential}

We continue with the isometries $V,W$ and environment spaces of
\cref{sec:amplification}. The preceding lemma gives approximations
whose error tends to zero at every rate above the low rate. We now
use this vanishing error to obtain exponential decay, allowing an
arbitrarily small further increase in the cost rate.

Choose one sufficiently accurate approximation $X$ at blocklength $k$
and keep $k$ fixed as the total blocklength grows. The exact auxiliary
map provides a representation of the residual $H=V^{\otimes k}-X$,
so that
\begin{equation*}
 (X+H)^{\otimes m}=V^{\otimes km}.
\end{equation*}
Thus the full expansion equals the target exactly. Truncating it
introduces only the omitted-tail error; there is no additional error
from repeating an approximate target, as in \eqref{eq:block-error}.
Since $\op H$ can be made arbitrarily small by choosing $k$ large
enough, we can retain terms with only a small fraction of residual
factors while keeping the tail exponentially small in $m$.
This limits the increase in cost rate. With $k$ fixed, exponential
decay in $m$ also gives exponential decay in the total blocklength.

\begin{lemma}[Fixed-block exponential amplification]\label{lem:exponential}
Suppose $V=\mathcal T_1(C_*)$ with $\op{C_*}\le2^{L/2}$ for
some $L\ge0$. At a rate $S\ge0$, assume that there exist auxiliary
maps $A_n$ for all sufficiently large $n$, with $\op{A_n}\le2^{nS/2}$ and
$\varepsilon_n=\op{\mathcal T_n(A_n)-V^{\otimes n}}\to0$.

Then for every $R>S$, there exist $K,\gamma>0$ such that, for all sufficiently
large $N$, there are auxiliary maps $C_N$ satisfying
\begin{equation}
 \op{C_N}\le2^{NR/2},\qquad
 \op{\mathcal T_N(C_N)-V^{\otimes N}}\le Ke^{-\gamma N}.
 \label{eq:exp-output}
\end{equation}
\end{lemma}

\begin{proof}
If $R\ge L$, the exact auxiliary maps suffice, so assume $S<R<L$.
Choose $\eta\in(0,1)$ with $\bar R=S+\eta(L-S)<R$, and set
$z=4^{1/\eta}$. Since $\varepsilon_n\to0$, there is a blocklength
$k$ large enough that $\varepsilon_k\le1/(1+z)$ and
$v:=\bar R+2\log3/k<R$. Keep $k$ fixed as $N\to\infty$.

Put $X=\mathcal T_k(A_k)$ and
$H=V^{\otimes k}-X=\mathcal T_k(C_*^{\otimes k}-A_k)$.
Then $\op X\le1+\varepsilon_k$, $\op H\le\varepsilon_k$,
and $(X+H)^{\otimes m}=V^{\otimes km}$ exactly. As in
\eqref{eq:truncation}, let $Z_m$ contain the terms containing at
most $\lfloor\eta m\rfloor$ factors $H$. The difference from
$V^{\otimes km}$ is the omitted sum of terms with $j>\eta m$
residual factors. There are $\binom mj$ terms with exactly $j$
such factors, each of norm at most
$(1+\varepsilon_k)^{m-j}\varepsilon_k^j$. Hence
\begin{align*}
 \op{Z_m-V^{\otimes km}}
 &\le\sum_{j>\eta m}\binom mj
              (1+\varepsilon_k)^{m-j}\varepsilon_k^j\\
 &\le z^{-\eta m}\sum_{j=0}^m\binom mj
              (1+\varepsilon_k)^{m-j}(z\varepsilon_k)^j\\
 &=z^{-\eta m}(1+\varepsilon_k+z\varepsilon_k)^m
 \le2^{-m}.
\end{align*}
For the second inequality, $z>1$ and $j>\eta m$ imply
$\varepsilon_k^j\le z^{-\eta m}(z\varepsilon_k)^j$.
We then extend the sum to all $0\le j\le m$, adding only
nonnegative terms, and apply the binomial theorem.
Finally, our choices $z=4^{1/\eta}$ and
$\varepsilon_k\le1/(1+z)$ give
$z^{-\eta}=1/4$ and $1+(1+z)\varepsilon_k\le2$.
The resulting bound is therefore $(\tfrac14\cdot2)^m=2^{-m}$.

Replace $X$ by $A_k$ and $H$ by $C_*^{\otimes k}-A_k$ in the
same truncated sum to obtain an auxiliary map $D_m$ with
$\mathcal T_{km}(D_m)=Z_m$. Since $S<L$, the auxiliary map
$C_*^{\otimes k}-A_k$ representing the correction has norm at most
$2\cdot2^{kL/2}$. Therefore
\begin{align*}
 \op{D_m}
 &\le\sum_{j\le\eta m}\binom mj2^j
          2^{k[(m-j)S+jL]/2}\\
 &\le3^m2^{km\bar R/2}=2^{kmv/2}.
\end{align*}
For $N=km+\ell$ with $0\le\ell<k$, set
$C_N=D_m\otimes C_*^{\otimes\ell}$. Exact padding preserves
the error, and
$\op{C_N}^2\le2^{kmv+\ell L}
=2^{Nv+\ell(L-v)}\le2^{NR}$ for all sufficiently large $N$:
the bounded remainder $\ell(L-v)$ is absorbed by $N(R-v)>0$.
Finally, $m=\lfloor N/k\rfloor$ implies
$2^{-m}\le2e^{-(\ln2)N/k}$. This proves \eqref{eq:exp-output}
with $K=2$ and $\gamma=(\ln2)/k$.
\end{proof}

We can now complete the approximation theorem. The rate margins are
chosen before the blocklength, so each application of amplification
uses constants independent of the total number of uses.

\begin{proof}[Proof of \cref{thm:approximation}]
Let $L$ be as in \cref{lem:finite}. \cref{cor:exact}
provides $C_*$ with $\op{C_*}\le2^{L/2}$.
For a target rate $S\ge L$, the auxiliary maps $C_*^{\otimes n}$
are exact and have the required norm. Otherwise $d<S<L$.
Fix rates $d<r<S_0<S$ and choose
$\delta\in(\sqrt{d/r},1)$. Combining \cref{lem:weak-tail}
with \cref{prop:environment} gives
the required constant-error approximation at rate $r$.
Apply \cref{lem:amplification} with
$V=V_N$, $W=V_M$, and target rate $S_0$ to obtain vanishing
error at that rate. \cref{lem:exponential}, with input rate
$S_0$ and target rate $S$, upgrades the error to an exponential.
Both lemmas cover every sufficiently large blocklength and preserve
the prescribed tensor-power environments. The exact case covers
$d=L$ as well, including $L=0$.
\end{proof}

\section{Testing and Smoothing Consequences}\label{sec:consequences}

The approximation theorem gives more than the fixed-error testing
limit. It also identifies the threshold of the optimized testing
tail and the asymptotic cost of approximating the null channel by a
subchannel dominated by the alternative. Throughout this section,
$\N,\M:A\to B$ are channels with $d=D^\infty(\N\Vert\M)<\infty$.

\subsection{A sharp testing threshold}\label{sec:hockey-threshold}

Below $d$, block repetition produces tests whose null acceptance tends
to one. Above $d$, the parallel converse forces that acceptance to
vanish whenever the testing objective is positive. These two facts
locate the threshold of the hockey-stick divergence.

\begin{corollary}[Hockey-stick threshold]\label{cor:threshold}
Let $\N,\M:A\to B$ be channels with $d=D^\infty(\N\Vert\M)<\infty$,
and set $E_n(r)=E_{2^{nr}}(\N^{\otimes n}\Vert\M^{\otimes n})$
for $r\ge0$, where $E_t$ is defined in \eqref{eq:hockey}. Then
\begin{equation*}
 \lim_{n\to\infty}E_n(r)=
 \begin{cases}
  1,&0\le r<d,\\
  0,&r>d.
 \end{cases}
\end{equation*}
For each fixed $r>d$, convergence to zero is exponential.
\end{corollary}

\begin{proof}
A positive objective value $p-2^{nr}q$ implies $q<2^{-nr}$ and
is at most $p$. For $r>d$, \cref{prop:parallel} therefore
bounds every such value by $K_r e^{-\gamma_r n}$ for all large $n$.
Taking the supremum proves the exponential upper bound.

For $0\le r<d$, choose a blocklength $k$ and a pure input whose
output states $\rho,\sigma$ satisfy $D(\rho\Vert\sigma)/k>r$.
Fix $u$ with $kr<u<D(\rho\Vert\sigma)$. The direct part of the
state Stein lemma gives tests on $m$ copies with $p_m\to1$ and
$q_m\le2^{-mu}$ for all sufficiently large $m$.
At total blocklength $n=mk+\ell$, $0\le\ell<k$, use the padding
construction from \eqref{eq:block-direct}. It follows that
$E_n(r)\ge p_m-2^{nr-mu}\to1$, because $mu/n\to u/k>r$.
Together with $E_n(r)\le1$, this proves the lower-rate limit.
\end{proof}

No conclusion is asserted at the critical rate $r=d$.

\subsection{Uniform approximation by dominated subchannels}\label{sec:subchannels}

A subchannel is a completely positive trace-nonincreasing map.
For a Hermiticity-preserving map $\Phi:\mathcal L(A)\to\mathcal L(B)$,
its diamond norm can be evaluated as
\[
 \dn\Phi=\sup_{\substack{\ket\psi\in R\otimes A\\\|\ket\psi\|=1}}
 \trn{(\id_R\otimes\Phi)(\ket\psi\bra\psi)},
\]
where $R\simeq A$ suffices~\cite{Watrous,WatrousSDP}. We use the unhalved
norm throughout. Thus this distance controls all inputs and references
simultaneously.

The approximating dilation need not be a contraction. A rescaling
makes it one while preserving the domination rate and changing the
error by only a constant factor.

\begin{corollary}[Dominated subchannels]\label{cor:subchannel}
Let $\N,\M:A\to B$ be channels with $d=D^\infty(\N\Vert\M)<\infty$.
For every $S>d$, there are $K,\gamma>0$ and subchannels
$\mathcal L_n:A^n\to B^n$, for all sufficiently large $n$, such that
\begin{equation}
 \mathcal L_n\cpo2^{nS}\M^{\otimes n},\qquad
 \dn{\mathcal L_n-\N^{\otimes n}}\le Ke^{-\gamma n}.
 \label{eq:subchannel}
\end{equation}
\end{corollary}

\begin{proof}
Fix dilations and take $C_n,e_n$ from
\cref{thm:approximation} at rate $S$. Put
$\widehat V_n=(I_{B^n}\otimes C_n)V_M^{\otimes n}$ and
$\widetilde V_n=\widehat V_n/(1+e_n)$.
Since $\op{\widehat V_n}\le1+e_n$, the operator
$\widetilde V_n$ is a contraction. It defines a CP map
$\mathcal L_n(X)=\Tr_{\mathsf E_N^{\otimes n}}
\widetilde V_nX\widetilde V_n^\dagger$. For $X\ge0$,
$\Tr\mathcal L_n(X)=\Tr X\widetilde V_n^\dagger\widetilde V_n
\le\Tr X$, so this map is a subchannel.

Let $F_{M,n}$ be the Kraus-column matrix of $V_M^{\otimes n}$ and
$D_n=(C_n/(1+e_n))^{\mathsf T}$. Then
\[
 J_{\mathcal L_n}=F_{M,n}D_nD_n^\dagger F_{M,n}^\dagger
 \le2^{nS}F_{M,n}F_{M,n}^\dagger
 =2^{nS}J_{\M^{\otimes n}}.
\]
This proves the CP domination. For the error, the identity
$(1+e_n)V_N^{\otimes n}-\widehat V_n
=e_nV_N^{\otimes n}+(V_N^{\otimes n}-\widehat V_n)$ gives
$\op{V_N^{\otimes n}-\widetilde V_n}
\le2e_n/(1+e_n)\le2e_n$.
For a pure input with any reference, let $x,y$ be the output vectors
of these two operators. They have norm at most one and
$\|x-y\|\le2e_n$. Decomposing their rank-one difference yields
\[
 \trn{\ket x\bra x-\ket y\bra y}
 \le(\|x\|+\|y\|)\|x-y\|\le4e_n.
\]
The partial trace is trace-norm contractive on Hermitian operators.
Taking that trace and optimizing over inputs therefore gives
$\dn{\mathcal L_n-\N^{\otimes n}}\le4e_n$.
Absorbing the factor four into the prefactor proves
\eqref{eq:subchannel}.
\end{proof}

We express the resulting rate through max-relative entropy. Following
the order-based definition for states~\cite{Datta} and its channel
extension~\cite{WangWilde}, for a nonzero CP map $\Phi$ and a channel
$\Psi$ on the same spaces set
$\Dmax(\Phi\Vert\Psi)=\log\inf\{t>0:\Phi\cpo t\Psi\}$,
with value $+\infty$ if no finite $t$ is feasible.
For $0<\varepsilon<1$, define
\begin{equation}
 \Dmax^{\varepsilon,\le}(\N\Vert\M)
 =\inf_{\substack{\mathcal L\text{ a subchannel}\\
                  \dn{\mathcal L-\N}\le\varepsilon}}
       \Dmax(\mathcal L\Vert\M).
 \label{eq:smoothed-dmax}
\end{equation}
The superscript $\le$ records that trace loss is allowed.
Since $\dn\N=1$, the condition $\varepsilon<1$ excludes the zero
map; finite values of \eqref{eq:smoothed-dmax} can nevertheless be
negative. This is a channel analogue of the smooth-entropy
asymptotics for states~\cite{TCR,DMHB}, with the normalization
constraint on the approximant treated separately~\cite{Gour}.

The dominated subchannels give the upper bound on the smoothing rate.
For the lower bound, closeness to the null channel forces every feasible
approximant to preserve the acceptance probability of a good test, up to the
smoothing error.
Together these observations give the following asymptotic identity.

\begin{corollary}[Subchannel asymptotic equipartition]\label{cor:aep}
Let $\N,\M:A\to B$ be channels with $d=D^\infty(\N\Vert\M)<\infty$.
For every $0<\varepsilon<1$,
\begin{equation}
 \lim_{n\to\infty}\frac1n
 \Dmax^{\varepsilon,\le}(\N^{\otimes n}\Vert\M^{\otimes n})=d.
 \label{eq:subchannel-aep}
\end{equation}
\end{corollary}

The converse follows from a one-shot estimate: for every $n\ge1$ and
$0<\eta<1-\varepsilon$, the quantity
$\Dmax^{\varepsilon,\le}(\N^{\otimes n}\Vert\M^{\otimes n})$
is at least $-\log\beta_n^*(\eta)+\log(1-\eta-\varepsilon)$.
We establish this bound as part of the proof.

\begin{proof}
Let $\mathcal L$ be a feasible subchannel in
\eqref{eq:smoothed-dmax}, and take any $t>0$ with
$\mathcal L\cpo t\M^{\otimes n}$. For any pure input $\ket\psi$ and
effect $T$, let $p,q$ be the acceptance probabilities under
$\N^{\otimes n},\M^{\otimes n}$. Diamond-norm closeness and CP
domination imply
\begin{equation*}
 p-\varepsilon
 \le\Tr T(\id_R\otimes\mathcal L)(\ket\psi\bra\psi)\le tq.
\end{equation*}
If $p\ge1-\eta$, then $q\ge(1-\eta-\varepsilon)/t$.
Taking the infimum over all such inputs and tests gives
$\beta_n^*(\eta)\ge(1-\eta-\varepsilon)/t$, and therefore
$\log t\ge-\log\beta_n^*(\eta)+\log(1-\eta-\varepsilon)$.
Taking infima over feasible $t$ and then over $\mathcal L$, with the
convention $\inf\varnothing=+\infty$, proves the one-shot estimate.
Fixing $\eta=(1-\varepsilon)/2$ and applying
\cref{prop:parallel} gives the asymptotic lower bound $d$.

For the upper bound, take any $S>d$.
\cref{cor:subchannel} gives, at all sufficiently large $n$,
a feasible subchannel with max-relative entropy at most $nS$.
Thus the limit superior is at most $S$. Let $S\downarrow d$ to
obtain \eqref{eq:subchannel-aep}.
\end{proof}

The same argument allows a varying smoothing radius. If its logarithmic
inverse is sublinear in $n$, the exponentially accurate approximants
are eventually feasible at that radius.

\begin{corollary}[Subexponential smoothing errors]
\label{cor:vanishing-smoothing}
Let $\N,\M:A\to B$ be channels with $d=D^\infty(\N\Vert\M)<\infty$.
Suppose $0<\varepsilon_n\le\varepsilon_*<1$ for a fixed
$\varepsilon_*$, and $\log(1/\varepsilon_n)=o(n)$. Then
\begin{equation}
 \lim_{n\to\infty}\frac1n
 \Dmax^{\varepsilon_n,\le}(\N^{\otimes n}\Vert\M^{\otimes n})=d.
 \label{eq:vanishing-smoothing-aep}
\end{equation}
\end{corollary}

\begin{proof}
Since $\varepsilon_n\le\varepsilon_*$, the feasible set at radius
$\varepsilon_n$ is contained in that at radius $\varepsilon_*$.
\cref{cor:aep} therefore gives the lower bound.
For any $S>d$, \cref{cor:subchannel} supplies subchannels
with domination $2^{nS}$ and error $Ke^{-\gamma n}$.
The assumption on $\varepsilon_n$ implies
$\ln(1/\varepsilon_n)=o(n)$; hence
$Ke^{-\gamma n}\le\varepsilon_n$ for all sufficiently large $n$.
These subchannels are feasible at radius $\varepsilon_n$, giving a
limit superior at most $S$. Letting $S\downarrow d$ proves the claim.
\end{proof}

The sequence of radii need not be monotone or converge to zero.
Trace loss, however, is essential to the class of approximants used
here. Rescaling gives a contraction, and completing the associated
subchannel to a trace-preserving map need not preserve the domination
rate. The asymptotic equipartition property with trace-preserving
channel smoothing in fact fails in general~\cite{Gour}.

\section{Discussion}\label{sec:discussion}

We have established $D^\infty(\N\Vert\M)$ as the fixed-error Stein
threshold under arbitrary adaptive control. The proof also gives a way
to pass from a weak testing bound to exponential accuracy without first
identifying the optimizing input states. Once the weak bound is
converted into a uniform environmental approximation, finite CP
domination supplies an exact representation of the residual. The tensor
constructions can then improve the error at an arbitrarily small
increase in cost rate. This explains how the same weak estimate
leads to the parallel strong converse, R\'enyi continuity, and the
adaptive converse.

The main quantitative question is how the approximation error depends
on the rate gap above $d$. In \cref{lem:exponential}, the error exponent
is $\gamma=(\ln2)/k$, where the blocklength $k$ is chosen after fixing
the rate gap. This proves positivity for every fixed gap, but the
construction does not optimize the required blocklength. Sharper control
of this dependence would give quantitative continuity bounds through
\eqref{eq:renyi-quantitative}, and hence explicit lower bounds on the
strong-converse exponent near the threshold. It would also give bounds
on the domination rate needed for a prescribed exponential smoothing
error, extending the regime covered by
\cref{cor:vanishing-smoothing}.

\section*{Acknowledgment}
OpenAI Codex was used in proof development, mathematical checking,
literature work, and preparation of the source manuscript, and in
expanding, reorganizing, and typesetting the present exposition.


\begingroup
\small
\raggedright

\endgroup
\end{document}